\documentclass[12pt]{article}

\usepackage{geometry}

\usepackage{hyperref}
\usepackage{inputenc}
\hypersetup{pdfencoding=unicode}
\inputencoding{utf8}
\makeatother

\usepackage{inputenc}
\usepackage[english]{babel}

\usepackage{amsfonts}
\usepackage{amsthm}
\usepackage{amsmath}
\usepackage[nottoc]{tocbibind}
\usepackage{graphicx}

\usepackage{float}
\usepackage{wrapfig}

\usepackage{lipsum}
\usepackage{tikz}
\usetikzlibrary{calc}

\graphicspath{{Pictures/}}

\newcommand{\tab}{\hspace*{2em}}
\providecommand{\tabularnewline}{\\}

\newlength{\localh}
\newlength{\locald}
\newbox\mybox
\def\mp#1#2{\setbox\mybox\hbox{#2}\localh\ht\mybox\locald\dp\mybox\addtolength{\localh}{-\locald}\raisebox{-#1\localh}{\box\mybox}}
\def\m#1{\scalebox{1}{\mp{0.5}{#1}}}

\newenvironment{qcircuit}{%
  \begingroup%
  \def\dot##1{\fill (##1) circle (.15);}%
  \def\triangle##1{\fill (##1)+(.3,0) --
    +(-.15,.25) -- +(-.15,-.25) -- cycle;}%
  \def\triangleadj##1{\fill (##1)+(-.3,0) --
    +(.15,.25) -- +(.15,-.25) -- cycle;}%
  \def\notgate##1{\filldraw[fill=white,thick] (##1) circle (.3);
                  \draw[thick] (##1)+(.3,0) -- +(-.3,0);
                  \draw[thick] (##1)+(0,.3) -- +(0,-.3);
                }%
  \def\controlled{\gencontrolled{\dot}}%
  \def\gencontrolled##1##2##3##4{\foreach\x in {##4} {
                    \draw[thick] (##3 |- 0,\x) -- (##3);
                  }
                  ##2{##3};
                  \foreach\x in {##4} {
                    ##1{##3 |- 0,\x};
                  }}%
  \def\gate##1##2{\biggate{##1}{##2}{##2}}%
  \def\biggate##1##2##3{\filldraw[fill=white,thick] (##2)+(-.45,-.4)
    rectangle ($(##3)+(.45,.4)$); \draw ($.5*(##2)+.5*(##3)$) node
    {\small ##1};}%
  \def\widebiggate##1##2##3##4{\filldraw[fill=white,thick] (##2)+(-1*##4,-.4)
    rectangle ($(##3)+(##4,.4)$); \draw ($.5*(##2)+.5*(##3)$) node
    {\small ##1};}%
  \def\widegate##1##2##3{\widebiggate{##1}{##3}{##3}{##2}}
  \def\gridx##1##2##3{\foreach\x in {##3} {
        \draw[thick] (##1,\x) -- (##2,\x);
      }}%
  \def\grid{\gridx{0}}%
  \def\cross##1{\draw (##1)+(-.3,-.3) -- +(.3,.3);
             \draw (##1)+(.3,-.3) -- +(-.3,.3);}%
  \def\bar##1{\fill (##1)+(-.05,-.4) rectangle ($(##1)+(.05,.4)$);}%
  \def\init##1##2{\bar{##2}; \leftlabel{##1}{##2};}
  \def\term##1##2{\bar{##2}; \rightlabel{##1}{##2};}
  \def\leftlabel##1##2{\draw (##2) node[left] {##1};}%
  \def\rightlabel##1##2{\draw (##2) node[right] {##1};}%
  \def\wirelabel##1##2{\draw (##2)+(0,-.15) node[above] {\footnotesize ##1};}%
  \def\multiwire##1{\draw (##1)+(-.1,-.3) -- +(.1,.3);}%
  \begin{tikzpicture}}
  {\end{tikzpicture}\endgroup}

\begin{document}

\begin{titlepage}

\newcommand{\HRule}{\rule{\linewidth}{0.5mm}}

\center

\textsc{\LARGE Dalhousie University}\\[1.5cm]
\textsc{\Large Honours Thesis}\\[0.5cm]

\HRule \\[0.4cm]
{ \huge \bfseries The Exact Synthesis of $1$- and $2$-Qubit Clifford+$T$ Circuits}\\[0.4cm]
\HRule \\[1cm]

\begin{minipage}{0.4\textwidth}
\begin{flushleft} \large
\emph{Author:}\\
Travis \textsc{Russell}
\end{flushleft}
\end{minipage}
~
\begin{minipage}{0.4\textwidth}
\begin{flushright} \large
\emph{Supervisor:} \\
Dr. Peter \textsc{Selinger}
\end{flushright}
\end{minipage}\\[1.5cm]

{\large May 2, 2014}\\[1.5cm]

\begin{abstract}
We describe a new method for the decomposition of an arbitrary $n$ qubit operator with entries in $\mathbb{Z}[i,\frac{1}{\sqrt{2}}]$, i.e., of the form $\frac{a+b\sqrt{2}+i(c+d\sqrt{2})}{\sqrt{2}^{k}}$, into Clifford+$T$ operators where $n\le 2$. This method achieves a bound of $O(k)$ gates using at most one ancilla using decomposition into $1$- and $2$-level matrices which was first proposed by Giles and Selinger in [2].
\end{abstract}

\vfill

\end{titlepage}

\tableofcontents

\newpage

\section{Introduction}

\subsection{Exact synthesis of Clifford+$T$ circuits}
In quantum information theory, the decomposition of unitary operators into gates from some fixed universal set is an important problem. Depending on the operator, this can be done either exactly, which is known as \textit{exact synthesis}, or approximately up to some given accuracy $\epsilon$, which is known as \textit{approximate synthesis}.\\
\tab In this paper we focus on the problem of exact synthesis for $1$- and $2$-qubit operators using the Clifford+$T$ gate set. It is known that the Clifford group on $n$ qubits, generated by the Hadamard gate $H$, the phase gate $S$, the controlled-not gate, and the scalar $\omega=e^{\frac{i\pi}{4}}=\frac{1+i}{\sqrt{2}}$, along with the non-Clifford operator $T$, forms a universal gate set [1].
\[
\omega=e^{\frac{i\pi}{4}}\quad H=\frac{1}{\sqrt{2}}\begin{pmatrix}1 & 1\\
1 & -1
\end{pmatrix}\quad S=\begin{pmatrix}1 & 0\\
0 & i
\end{pmatrix}
\]
\[
CNOT=\begin{pmatrix}1 & 0 & 0 & 0\\
0 & 1 & 0 & 0\\
0 & 0 & 0 & 1\\
0 & 0 & 1 & 0
\end{pmatrix}\quad T=\begin{pmatrix}1 & 0\\
0 & \omega
\end{pmatrix}
\]

\subsection{Related work}
Recently, there has been some interest in finding an efficient algorithm for the exact synthesis of $n$-qubit operators. Giles and Selinger first presented an algorithm for the exact synthesis of $n$-qubit operators with entries in the ring $\mathbb{Z}[\frac{1}{\sqrt{2}},i]$ using a method involving decomposition into $1$- and $2$-level matrices, in which one ancilla is sufficient [2]. If the matrix entries are of the form
$\frac{a+b\sqrt{2}+i(c+d\sqrt{2})}{\sqrt{2}^k}$, then this algorithm
achieves an upper bound of $O(3^{2^{n}}nk)$ gates, which is far from optimal.\\
\tab The problem with this algorithm was that the exponent in the denominator of the remaining entries of the matrix may increase after applying the $1$- and $2$-level matrices, and this increases the number of operations needed to reduce the matrix. Kliuchnikov then introduced an algorithm which achieved an upper bound of $O(4^{n}nk)$ gates, using a different method requiring at most two ancillas [3].\\
\tab With the revelation that this efficient bound was possible, this paper presents an algorithm that achieves an efficient bound for the $1$- and $2$-qubit cases while using the original method of decomposition into $1$- and $2$-level operators in which one ancilla is sufficient.

\section{Some algebra}

We will begin by defining some notation and terminology.


\subsection{Some rings} 

Recall that $\mathbb{N}$ is the set of all natural numbers including $0$ and
$\mathbb{Z}$ is the ring of integers. Let $\omega=e^{\frac{i\pi}{4}}=\frac{(1+i)}{\sqrt{2}}$.
Let $\mathbb{D}$ be the ring of \textit{dyadic fractions}, defined as $\mathbb{D}=\{\frac{a}{2^{n}}\mid a\in\mathbb{Z},\:n\:\epsilon\:\mathbb{N}\}$.

\newtheorem{mydef1}{Definition}
\begin{mydef1}
\rm{$\mathbb{Z}[\omega]=\{a\omega^{3}+b\omega^{2}+c\omega+d\mid a,\:b,\:c,\:d\:\epsilon\:\mathbb{Z}\}$  and $\mathbb{D}[\omega]=\{a\omega^{3}+b\omega^{2}+c\omega+d\mid a,\:b,\:c,\:d\:\epsilon\:\mathbb{D}\}$
are subrings of the complex numbers and have addition and multiplication
defined by the ring axioms as well as the property $\omega^{4}=-1$.}
\end{mydef1}
We note that $\mathbb{Z}[\frac{1}{\sqrt{2}},i]=\mathbb{D}[\omega]$.


\subsection{Conjugate and norm} 

\newtheorem{defconj}[mydef1]{Definition}
\begin{defconj}
\rm{Because $\omega$ is a primitive $8^{th}$ root of unity, $\mathbb{Z}[\omega]$
has $\phi(8)=4$ automorphisms. One such automorphism is the usual
complex conjugation which maps $i$ to $-i$ and $\sqrt{2}$ to itself.
We will denote complex conjugation by $(-)^{\dagger}$. For any element
in $\mathbb{Z}[\omega]$, we have
\[
(a\omega^{3}+b\omega^{2}+c\omega+d)^{\dagger}=-c\omega^{3}+b\omega^{2}-a\omega+d.
\]
Another automorphism is $\sqrt{2}$-conjugation, which maps $i$ to
itself and $\sqrt{2}$ to $-\sqrt{2}$. We will denote $\sqrt{2}$-conjugation
by $(-)^{\bullet}$. For any element in $\mathbb{Z}[\omega]$, we
have
\[
(a\omega^{3}+b\omega^{2}+c\omega+d)^{\dagger}=-a\omega^{3}+b\omega^{2}-c\omega+d.
\]
$ $The remaining two automorphisms are obviously the identity function
and $(-)^{\dagger\bullet}=(-)^{\bullet\dagger}$.}
\end{defconj}

\newtheorem{defnorm}[mydef1]{Definition}
\begin{defnorm}
\rm{We define a ring norm for $\mathbb{Z}[\omega]$ and $\mathbb{D}[\omega]$. It is given by the following formula for $t=a\omega^{3}+b\omega^{2}+c\omega+d$:
\[
\mathcal{N}(t)=tt^{\dagger}t^{\bullet}t^{\dagger\bullet}=(a^{2}+b^{2}+c^{2}+d^{2})^{2}-2(cd+bc+ab-da)^{2}
\]}
\end{defnorm}


\subsection{Euclidean domains} 

\newtheorem{remarkdomain}{Remark}
\begin{remarkdomain}
\rm{We note that $\mathbb{Z}[\omega]$ is a Euclidean domain, with a Euclidean function given by $\left|\mathcal{N}(t)\right|$. }
\end{remarkdomain}

As usual, we write $t|s$ if $t$ divides $s$, i.e., if there exists an $r$ such that $rt=s$. If $r$ is a unit of the ring, then we say $t$ and $s$ are \textit{associates} and we denote this by $t\sim s$. This means that $t|s$ and $s|t$.

\newtheorem{remark1}[remarkdomain]{Remark}
\begin{remark1}
\rm{Let $\delta=1+\omega\in\mathbb{Z}[\omega]$. Then $\delta^{2}\sim\sqrt{2}$ and $\delta^{2k}\sim\sqrt{2}^{k}$. This implies that any $t\in\mathbb{Z}[\omega]$ can be written as $ $$t=\frac{a\omega^{3}+b\omega^{2}+c\omega+d}{\delta^{k}}$ with $a,\: b,\: c,\: d\in\mathbb{Z}$,
$k\in\mathbb{N}$.}

\end{remark1}

\subsection{Quotient mappings}

\newtheorem{mydef2}[mydef1]{Definition}

\begin{mydef2}
\rm{Let $n\geq0$. Recall that $(\delta^{n})$ is the ideal generated by $\delta^{n}$, i.e., $(\delta^n)$ is the set of all $t\in\mathbb{Z}[\omega]$ that are divisible by $\delta^n$. Let $\rho_{n}:\mathbb{Z}[\omega]\rightarrow\mathbb{Z}[\omega]/\left(\delta^{n}\right)$
be the standard quotient mapping $x\mapsto x+\left(\delta^{n}\right)$.}\end{mydef2}

We note that the elements of $\mathbb{Z}[\omega]/\left(\delta^{n}\right)$ are the the equivalence classes of elements of $\mathbb{Z}[\omega]\pmod{\delta^{n}}$.

\newtheorem{remark2}[remarkdomain]{Remark}
\begin{remark2}
\rm{For $n\geq0$, $\mathbb{Z}[\omega]/\left(\delta^{n}\right)$ has $2^{n}$ elements. For example,
\[
\mathbb{Z}[\omega]/\left(\delta\right)=\{0,\:1\},
\]
\[
\mathbb{Z}[\omega]/\left(\delta^{2}\right)=\{0,\:1,\:\omega,\:1+\omega\},
\]
\[
\mathbb{Z}[\omega]/\left(\delta^{3}\right)=\{0,\:1,\:\omega,\:\omega^{2},\:\omega^{3},\:1+\omega,\:1+\omega^{2},\:1+\omega^{3}\}.
\]}
\end{remark2}

\subsection{Denominator exponents}
\newtheorem{mydef3}[mydef1]{Definition}
\begin{mydef3}
\rm{Let $t\in\mathbb{D}[\omega]$. A natural number $k\in\mathbb{N}$
is called a \textit{$\delta$-exponent} for $t$ if $\delta^{k}t\in\mathbb{Z}[\omega]$.
From Remark 2, it is obvious such a $k$ exists, and the
least such $k$ is called the \textit{least $\delta$-exponent} for $t$.}
\end{mydef3}
When dealing with a vector or matrix $U$, $k$ is a $\delta$-exponent for $U$ if it is a $\delta$-exponent for all of its entries. The least $\delta$-exponent of $U$ is thus the least $k$ that is a $\delta$-exponent for all of its entries.

\subsection{Residues}

\newtheorem{mydef4}[mydef1]{Definition}
\begin{mydef4}
\rm{Let $t\in\mathbb{D}[\omega]$, and let $k$ be a (not necessarily least) $\delta$-exponent for $t$. The \textit{$(n,k)$-residue} of $t$, in symbols $\rho_{n}^{k}(t)$, is defined to be}
\[
\rho_{n}^{k}(t)=\rho_{n}(\delta^{k}t).
\]
Similarly for a matrix $U$ with entries in $\mathbb{D}[\omega]$, we let $\rho_{n}^{k}(U)$ signify the matrix made up of the $(n,k)$-residues of the entries of $U$.
\end{mydef4}

\subsection{Reducibility}

\newtheorem{mydef5}[mydef1]{Definition}
\begin{mydef5}
\rm{We say $x\in\mathbb{Z}[\omega]$ is \textit{reducible}
if $\delta|x$.}
\end{mydef5}

\newtheorem{remark3}[remarkdomain]{Remark}

\begin{remark3}
\rm{The following are equivalent for $x\in\mathbb{Z}[\omega]$:

(a) $x$ is reducible;

(b) $\rho_{1}(x)=0$;

(c) $\rho_{3}(x)\in\{0,\:1+\omega,\:1+\omega^{2},\:1+\omega^{3}\}$.}
\end{remark3}

\section{Decomposition into elementary matrices}

\subsection{Elementary matrices}

\newtheorem{2leveldef}[mydef1]{Definition}
\begin{2leveldef}
\rm{The elementary matrices $\omega_{[j]}$, $H_{[j,m]}$, and $X_{[j,m]}$ are defined to be

\[
\omega_{[j]}=
 \bordermatrix{&\ldots&j&\ldots\cr
            \vdots&I&0&0\cr
            j&0&\omega&0\cr
            \vdots&0&0&I\cr
            }
\quad
H_{[j,m]}=
 \bordermatrix{&\ldots&j&\ldots&m&\ldots\cr
            \vdots&I&0&0&0&0\cr
            j&0&\frac{1}{\sqrt{2}}&0&\frac{1}{\sqrt{2}}&0\cr
            \vdots&0&0&I&0&0\cr
            m&0&\frac{1}{\sqrt{2}}&0&-\frac{1}{\sqrt{2}}&0\cr
            \vdots&0&0&0&0&I\cr
            }
\]

\[
X_{[j,m]}=
 \bordermatrix{&\ldots&j&\ldots&m&\ldots\cr
            \vdots&I&0&0&0&0\cr
            j&0&0&0&1&0\cr
            \vdots&0&0&I&0&0\cr
            m&0&1&0&0&0\cr
            \vdots&0&0&0&0&I\cr
            }.
\]
Note that $\omega_{[j]}$ is a $1$-level matrix and $H_{[j,m]}$, and $X_{[j,m]}$ are $2$-level matrices.}
\end{2leveldef}

\subsection{Properties of unitary matrices}

\newtheorem{lemma1}{Lemma}
\begin{lemma1}
Let $U$ be an $n\times n$ unitary matrix with entries in $\mathbb{D}[\omega]$ and least $\delta$-exponent $k$. Then if $k>0$, there is an even number of 1's in each row and each column of $\rho_{1}^{k}(U)$.
\end{lemma1}
\begin{proof}
Let $u_{jm}$ denote the element of $U$ in the $j$-th row and $m$-th column, and let $v_{jm} = \delta^{k}u_{jm}$. Then, because $U$ is unitary,
$u_{11}^{\dagger}u_{11}+u_{21}^{\dagger}u_{21}+...+u_{n1}^{\dagger}u_{n1}=1$,
which implies that $\frac{1}{\left(\delta^{\dagger}\delta\right)^{k}}(v_{11}^{\dagger}v_{11}+v_{21}^{\dagger}v_{21}+...+v_{n1}^{\dagger}v_{n1})=1$, so $v_{11}^{\dagger}v_{11}+v_{21}^{\dagger}v_{21}+...+v_{n1}^{\dagger}v_{n1}=\left(\delta^{\dagger}\delta\right)^{k}\equiv0\pmod{\delta}$.
Therefore an even number
of $v_{11},...,v_{n1}$ must be congruent to $1\pmod{\delta}$, and the lemma follows.
\end{proof}

\newtheorem{evennumb1s}[lemma1]{Lemma}
\begin{evennumb1s}
Let $U$ be an $n\times n$ unitary matrix with entries in $\mathbb{D}[\omega]$ and least $\delta$-exponent $k$. Then if $k>0$, any two distinct rows of $\rho_{1}^{k}(U)$ will have an even number of 1's in common.
\end{evennumb1s}
\begin{proof}
Assume that two distinct rows of $\rho_{1}^{k}(U)$ do not have an even number of 1's in common. Then, the inner product of these two rows will be congruent to $1\pmod{\delta}$, a contradiction because the inner product of any two rows must be congruent to $0\pmod{\delta}$. The lemma follows.
\end{proof}

\newtheorem{lemma2}[lemma1]{Lemma}
\begin{lemma2}
If $t\in\mathbb{D}[\omega]$ and $t^{\dagger}t\leq1$
and $(t^{\dagger}t)^{\bullet}\leq1$, then the least denominator exponent
of $t$ is not $1$.
\end{lemma2}

\begin{proof}
If $k=0$, then there is nothing to show, so assume that $k\ge1$.
Then we can write $t=\frac{a\omega^{3}+b\omega^{2}+c\omega+d}{\delta}$. Assume $t^{\dagger}t\le1$ and $(t^{\dagger}t)^{\bullet}\le1$.
Then, we have $\delta^{\dagger}\delta t^{\dagger} t =(a^{2}+b^{2}+c^{2}+d^{2})+(cd+bc+ab-da)\sqrt{2}\le\delta^{\dagger}\delta=2+\sqrt{2}$
and $(\delta^{\dagger}\delta t^{\dagger} t)^{\bullet}=(a^{2}+b^{2}+c^{2}+d^{2})-(cd+bc+ab-da)\sqrt{2}\le(\delta^{\dagger}\delta)^{\bullet}=2-\sqrt{2}$. Averaging the two equations, we get that $a^2 + b^2 + c^2 + d^2 \le2$. By a simple case distinction, one can see that the only solutions are $a=b=c=d=0$; $a,\:b=\pm1,\: c=d=0$; $b,\:c=\pm1,\: a=d=0$; $c,\:d=\pm1,\: a=b=0$; and $a,\:d=\pm1,\: b=c=0$. All these solutions
are reducible, so we must have $k>1$, and we are done.
\end{proof}

\newtheorem{corollary1}{Corollary}
\begin{corollary1}
For any unitary matrix $U$ with entries in $\mathbb{D}[\omega]$, the least $\delta$-exponent of $U$ is not $1$.
\end{corollary1}

\begin{proof}
This result follows directly from the previous lemma, as all entries of $U$ satisfy the required property.
\end{proof}

\subsection{Properties of \texorpdfstring{$\mathbb{Z}[\omega]/(\delta^{n})$}{ℤ[ω]/(δⁿ)}}

\newtheorem{denomexponent}[lemma1]{Lemma}
\begin{denomexponent}
Let $\begin{pmatrix}a\\
b\end{pmatrix}$ be a column vector with $a,\:b\in\mathbb{Z}[\omega]$ and $a\equiv b\pmod{\delta^{2}}$.
Then, the entries of $H_{[1,2]}\begin{pmatrix}a\\
b\end{pmatrix}$ are in $\mathbb{Z}[\omega]$.
\end{denomexponent}
\begin{proof}
We see that
\[
H_{[1,2]}\begin{pmatrix}a\\
b
\end{pmatrix}=\begin{pmatrix}\frac{a+b}{\sqrt{2}}\\
\frac{a-b}{\sqrt{2}}
\end{pmatrix}.
\]
Because $a\equiv b\pmod{\delta^{2}}$, this means $\delta^{2}\mid a-b$.
Because $\delta^{2}\sim\sqrt{2}$, we then have $\sqrt{2}\mid a-b$.
We also know that $\sqrt{2}\mid2b$, thus $\sqrt{2}|a+b$ and we are done.
\end{proof}

\newtheorem{corollarydenomexpon}[corollary1]{Corollary}
\begin{corollarydenomexpon}
Let $U$ be an $n\times n$ matrix with entries in $\mathbb{D}[\omega]$ and least $\delta$-exponent $k>1$. If the $j$-th and $m$-th row of $\rho_{2}^{k}(U)$ are equal, $H_{[j,m]}U$ has least $\delta$-exponent $k'\le k$.
\end{corollarydenomexpon}

\begin{proof}
This follows from the previous lemma.
\end{proof}

\newtheorem{corollarydenomexpon3}[lemma1]{Lemma}
\begin{corollarydenomexpon3}
Let $\begin{pmatrix}a\\
b\end{pmatrix}$ be a column vector with $a,\:b\in\mathbb{Z}[\omega]$ and $a\equiv b\pmod{\delta^{3}}$.
Then, the entries of $H_{[1,2]}\begin{pmatrix}a\\
b\end{pmatrix}$ are in $\mathbb{Z}[\omega]$ and are divisible by $\delta$.
\end{corollarydenomexpon3}

\begin{proof}
We see that
\[
H_{[1,2]}\begin{pmatrix}a\\
b
\end{pmatrix}=\begin{pmatrix}\frac{a+b}{\sqrt{2}}\\
\frac{a-b}{\sqrt{2}}
\end{pmatrix}.
\]
Because $a\equiv b\pmod{\delta^{3}}$, this means $\delta^{3}|a-b$. Because $\delta^{3}\sim\delta\sqrt{2}$, we then have $\delta\sqrt{2}|a-b$. We also know that $\delta\sqrt{2}|2b$, thus $\delta\sqrt{2}|a+b$ and we are done.
\end{proof}

\newtheorem{corollarydenomexpon1}[corollary1]{Corollary}
\begin{corollarydenomexpon1}
Let $U$ be an $n\times n$ matrix with entries in $\mathbb{D}[\omega]$ and least $\delta$-exponent $k>1$. If the $j$-th and $m$-th row of $\rho_{3}^{k}(U)$ are equal, then the $j$-th and $m$-th rows of $H_{[j,m]}U$ have least $\delta$-exponent $k'<k$.
\end{corollarydenomexpon1}

\begin{proof}
This follows from the previous lemma.
\end{proof}

\newtheorem{remark6}[remarkdomain]{Remark}
\begin{remark6}
\rm{When working with exponents of $\omega$ in $\mathbb{Z}[\omega]/(\delta^{3})$, we work modulo $4$. This is because $\omega^{4}\equiv 1\pmod{\delta^{3}}$.}
\end{remark6}

\newtheorem{congruent02}[lemma1]{Lemma}
\begin{congruent02}
If $\omega^{x}+\omega^{y}\equiv0\pmod{\delta^{3}}$ then $x\equiv y\pmod{4}$.
\end{congruent02}
\begin{proof}
We note that $\omega^{x}+\omega^{y}\equiv0\pmod{\delta^{3}}$ if and only if $\omega^{x}\equiv\omega^{y}\pmod{\delta^{3}}$. The claim then becomes obvious by Remark 3.
\end{proof}

\newtheorem{congruent04}[lemma1]{Lemma}
\begin{congruent04}
If $\omega^{x}+\omega^{y}+\omega^{z}+\omega^{u}\equiv0\pmod{\delta^{3}}$, then there are three cases up to permutations of $x$, $y$, $z$, and $u$:\\
(a)\: $x\equiv y\equiv z\equiv u\pmod{4}$;\\
(b)\: $x\equiv y\pmod{4}$, $z\equiv u\pmod{4}$, and $x\not\equiv z\pmod{4}$;\\
(c)\: $x$, $y$, $z$, and $u$ are distinct modulo 4.
\end{congruent04}
\begin{proof}
If $x$, $y$, $z$, and $u$ are distinct, we are in case (c) and are done. Otherwise, two are equal, say $x\equiv y\pmod{4}$. Then, we have $\omega^{x}+\omega^{y}\equiv\omega^{z}+\omega^{u}\equiv0\pmod{\delta^{3}}$, so by Lemma 6, $z\equiv u\pmod{4}$. Then, case (a) holds if $z\equiv x\pmod{4}$, and case (b) holds if $z\not\equiv x\pmod{4}$.
\end{proof}

\newtheorem{remark7}[remarkdomain]{Remark}
\begin{remark7}
\rm{Let $U$ be an $n\times n$ unitary matrix with entries in $\mathbb{D}[\omega]$ and (not necessarily least) $\delta$-exponent $k$. We introduce the following notation for $\rho_{m}^{k}(U)$:
\[
\rho_{m}^{k}(U)=A_{0}+\delta A_{1}+\delta^{2}A_{2}+\ldots+\delta^{m-1}A_{m-1},
\]
where $A_{i}$ is an $n\times n$ matrix with entries in $\{0,\:1\}$, noting that $\{1,\:\delta,\:\delta^{2}\,\ldots,\:\delta^{m-1}\}$ forms a basis for $\mathbb{Z}[\omega]/(\delta^{m})$ over $\{0,\:1\}$. For example, when $m=3$ we have:
\[
\begin{tabular}{|c|c|}
\hline
$0$ & $0+0\delta+0\delta^{2}$\tabularnewline
\hline
$1+\omega$ & $0+1\delta+0\delta^{2}$\tabularnewline
\hline
$1+\omega^{2}$ & $0+0\delta+1\delta^{2}$\tabularnewline
\hline
$1+\omega^{3}$ & $0+1\delta+1\delta^{2}$\tabularnewline
\hline
$1$ & $1+0\delta+0\delta^{2}$\tabularnewline
\hline
$\omega$ & $1+1\delta+0\delta^{2}$\tabularnewline
\hline
$\omega^{2}$ & $1+0\delta+1\delta^{2}$\tabularnewline
\hline
$\omega^{3}$ & $1+1\delta+1\delta^{2}$\tabularnewline
\hline
\end{tabular}
\]}
\end{remark7}

\subsection{Statement of the base case and main lemma}

\newtheorem{basecase}[lemma1]{Lemma}
\begin{basecase}
Let $n\le 4$, and let $U$ be an $n\times n$ unitary matrix with entries in $\mathbb{D}[\omega]$ and least $\delta$-exponent $k=0$. Then there exists a sequence $S_{1},\ldots, S_{m}$ of elementary operators such that $S_{1}\ldots S_{m}U=I$. Moreover, there exists a fixed bound $M$, depending only on $n$, such that $m\le M$.
\end{basecase}

\begin{proof}
The entries of $U$ satisfy $tt^{\dagger}\le1$ and $(tt^{\dagger})^{\bullet}\le1$. From this it follows that each entry is $0$ or $\omega^{l}$, and there is exactly one entry of $\omega^{l}$ in each row and each column. It is then trivial to reduce $U$ to $I$ using elementary operators of types $X$ and $\omega$.
\end{proof}

\newtheorem{maintheorem}[lemma1]{Lemma}
\begin{maintheorem}
Let $n\le4$, and let $U$ be an $n\times n$ unitary matrix with entries in $\mathbb{D}[\omega]$ and least $\delta$-exponent $k>1$. Then there exists two sequences $S_{1},\dots,S_{m}$ and $T_{1},\dots,T_{j}$ of elementary matrices such that $S_{1},\dots,S_{m}UT_{1},\dots,T_{j}$ has least $\delta$-exponent $k'<k$. Moreover, there exists a fixed bound $N$ dependent only on $n$ such that $m+j\le N$.
\end{maintheorem}

\subsection{Proof for \texorpdfstring{$2\times2$}{2 × 2} matrices}

We begin with the $2\times2$ case because the $1\times1$ case is trivial.
Let $U$ be a $2\times2$ unitary matrix with entries in $\mathbb{D}[\omega]$ and least $\delta$-exponent $k>1$. As a consequence of Lemma 1, we must have
\[\rho_{1}^{k}(U)=\begin{pmatrix}1 & 1\\
1 & 1
\end{pmatrix}.
\]
This means that
\[
\rho_{3}^{k}(U)=\begin{pmatrix}\omega^{a} & \omega^{b}\\
\omega^{c} & \omega^{d}
\end{pmatrix}
\]
for some $a$, $b$, $c$, $d\in\mathbb{Z}$.

\newtheorem{lemma2x2}[lemma1]{Lemma}
\begin{lemma2x2}
There exists an $x\pmod{4}$ such that $a+x\equiv c\pmod{4}$ and such
that $b+x\equiv d\pmod{4}$.
\end{lemma2x2}

\begin{proof}
Taking the inner product of the first and second rows, we get that $\omega^{c-a}+\omega^{d-b}\equiv0\pmod{4}$. By Lemma 6, this occurs only when $c-a\equiv d-b\pmod{4}$. Taking $x\equiv c-a\pmod{4}$, we are done.
\end{proof}

We now have
\[
\rho^k_3(\omega_{[1]}^{x}U)=\begin{pmatrix}\omega^{c} & \omega^{d}\\
\omega^{c} & \omega^{d}
\end{pmatrix}
\]
and by Corollary 3, applying the row operation $H_{[1,2]}$ reduces the least $\delta$-exponent, so we are done.

\subsection{Proof for \texorpdfstring{$3\times3$}{3 × 3} matrices}

Let $U$ be a $3\times3$ unitary matrix with entries in $\mathbb{D}[\omega]$ and least $\delta$-exponent $k>1$. As a consequence of Lemma 1, we have two cases up to permutations of rows and columns:
\[
(i)\:\rho_{1}^{k}(U)=\begin{pmatrix}1 & 1 & 0\\
1 & 1 & 0\\
0 & 0 & 0
\end{pmatrix}
\quad (ii)\:\rho_{1}^{k}(U)=\begin{pmatrix}1 & 1 & 0\\
1 & 0 & 1\\
0 & 1 & 1
\end{pmatrix}
\]

We note that case (ii) cannot occur by Lemma 2.
Thus, without loss of generality and up to permutations of rows and columns, we must have:
\[
\rho_{3}^{k}(U)=\begin{pmatrix}\omega^{a} & \omega^{b} & \delta x_{13}\\
\omega^{c} & \omega^{d} & \delta x_{23}\\
\delta x_{31} & \delta x_{32} & \delta x_{33}
\end{pmatrix}
\]
where $x_{jm}\in\mathbb{Z}[\omega]/(\delta^{2})$.
We begin by looking for potential values of $\delta x_{13}$. Because $U$ is unitary, we know that $\omega^{-a}\omega^{a}+\omega^{-b}\omega^{b}+\delta^{\dagger}\delta x_{13}^{\dagger}x_{13}\equiv 0\pmod{\delta^{3}}$, which implies that $\delta^{\dagger}\delta x_{13}^{\dagger}x_{13}\equiv 0\pmod{\delta^{3}}$. This implies that we must have $x^{\dagger}_{13}x_{13}\equiv0\pmod{\delta}$, hence $\delta|x_{13}$. Similarly, we have that $\delta$ divides $x_{23}$, $x_{31}$, and $x_{32}$.

\newtheorem{lemma4}[lemma1]{Lemma}
\begin{lemma4}
There exists an $x\pmod{4}$ such that $a+x\equiv c\pmod{4}$ and such
that $b+x\equiv d\pmod{4}$.
\end{lemma4}

\begin{proof}
Because $U$ is unitary, we know that $\omega^{-a}\omega^{c}+\omega^{-b}\omega^{d}+\delta^{\dagger}\delta x_{13}x_{23}\equiv0\pmod{\delta^{3}}$. But then $\omega^{-a}\omega^{c}+\omega^{-b}\omega^{d}\equiv0\pmod{\delta^{3}}$ because we know $\delta$ divides both $x_{13}$ and $x_{23}$. By Lemma 6, we have $c-a \equiv d-b\pmod{4}$, and setting
$x\equiv c-a\pmod{4}$ we are done.
\end{proof}

We now have
\[
\rho_{3}^{k}(\omega^{x}_{[1]}U)=\begin{pmatrix}\omega^{c} & \omega^{d} & \delta y_{13}\\
\omega^{c} & \omega^{d} & \delta x_{23}\\
\delta x_{31} & \delta x_{32} & \delta x_{33}
\end{pmatrix}
\]
where $y_{13}\equiv\omega^{x} x_{13}\pmod{\delta^{3}}$. Note that $\delta|y_{13}$.

\newtheorem{lemma3x3equal}[lemma1]{Lemma}
\begin{lemma3x3equal}
We always have $\delta y_{13}\equiv\delta x_{23}\pmod{\delta^3}$.
\end{lemma3x3equal}

\begin{proof}
By taking the inner product of the first column and the third column, we have $\omega^{-c}\delta y_{13} + \omega^{-c}\delta x_{23} +
\delta^{\dagger}\delta x_{31}^{\dagger} x_{33} \equiv
\omega^{-c}\delta(y_{13}+x_{23})\equiv 0\pmod{\delta^3}$. Hence,
since $\omega$ is invertible, $\delta(y_{13}+x_{23})\equiv 0\pmod{\delta^{3}}$, or equivalently, $\delta y_{13}\equiv \delta x_{23}\pmod{\delta^{3}}$.
\end{proof}

We now have
\[
\rho_{3}^{k}(\omega^{x}_{[1]}U)=\begin{pmatrix}\omega^{c} & \omega^{d} & \delta y_{13}\\
\omega^{c} & \omega^{d} & \delta y_{13}\\
\delta x_{31} & \delta x_{32} & \delta x_{33}
\end{pmatrix}
\]
and by Corollary 3, we
can apply the row operation $H_{[1,2]}$ and we are done.
\subsection{Proof for \texorpdfstring{$4\times4$}{4 × 4} matrices}

Let $U$ be a unitary $4\times4$ matrix with entries in $\mathbb{D}[\omega]$ and least $\delta$-exponent $k>1$. As a consequence of Lemmas 1 and 2, we have the following cases for $U$ up to permutations of rows and columns as well as taking transposes:

\[
(i)\:\rho_{1}^{k}(U)=\begin{pmatrix}1 & 1 & 0 & 0\\
1 & 1 & 0 & 0\\
0 & 0 & 0 & 0\\
0 & 0 & 0 & 0
\end{pmatrix}\quad(ii)\:\rho_{1}^{k}(U)=\begin{pmatrix}1 & 1 & 1 & 1\\
1 & 1 & 1 & 1\\
0 & 0 & 0 & 0\\
0 & 0 & 0 & 0
\end{pmatrix}\quad(iii)\:\rho_{1}^{k}(U)=\begin{pmatrix}1 & 1 & 0 & 0\\
1 & 1 & 0 & 0\\
0 & 0 & 1 & 1\\
0 & 0 & 1 & 1
\end{pmatrix}
\]
\[
(iv)\:\rho_{1}^{k}(U)=\begin{pmatrix}1 & 1 & 0 & 0\\
1 & 1 & 0 & 0\\
1 & 1 & 1 & 1\\
1 & 1 & 1 & 1
\end{pmatrix}\quad
(v)\:\rho_{1}^{k}(U)=\begin{pmatrix}1 & 1 & 1 & 1\\
1 & 1 & 1 & 1\\
1 & 1 & 1 & 1\\
1 & 1 & 1 & 1
\end{pmatrix}
\]
\noindent
We will now prove each of the above cases satisfies Theorem 1.
\subsubsection{Case (i)}
In case (i), we have
\[
\rho_{3}^{k}(U)=\begin{pmatrix}\omega^{a} & \omega^{b} & \delta x_{13} & \delta x_{14}\\
\omega^{c} & \omega^{d} & \delta x_{23} & \delta x_{24}\\
\delta x_{31} & \delta x_{32} & \delta x_{33} & \delta x_{34}\\
\delta x_{41} & \delta x_{42} & \delta x_{43} & \delta x_{44}
\end{pmatrix}.
\]
\newtheorem{lemmacasei}[lemma1]{Lemma}
\begin{lemmacasei}
If $U$ is of case (i), there exists an $x\pmod{4}$ such that $a+x\equiv c\pmod{4}$ and such
that $b+x\equiv d\pmod{4}$.
\end{lemmacasei}
\begin{proof}
Taking the inner product of the first column with itself, we have $\omega^{-a}\omega^{a}+\omega^{-c}\omega^{c}+\delta^{\dagger}\delta x_{31}^{\dagger}x_{31}+\delta^{\dagger}\delta x_{41}^{\dagger}x_{41}\equiv0\pmod{\delta^{3}}$.
We then have that $x_{31}^{\dagger}x_{31}+x_{41}^{\dagger}x_{41}\equiv0\pmod{\delta}$,
which can only occur if $x_{31}\equiv x_{41}\pmod{\delta}$. Using
a similar argument on the second column, we get that $x_{32}\equiv x_{42}\pmod{\delta}$.
Then, taking the inner products of the first and second columns, we get
$\omega^{-a}\omega^{b}+\omega^{-c}\omega^{d}+\delta^{\dagger}\delta x_{31}^{\dagger}x_{32}+\delta^{\dagger}\delta x_{41}^{\dagger}x_{42}\equiv0\pmod{\delta^{3}}$.
Because $x_{31}^{\dagger}x_{32}\equiv x_{41}^{\dagger}x_{42}\pmod{\delta}$,
this implies that $\delta^{\dagger}\delta x_{31}^{\dagger}x_{32}+\delta^{\dagger}\delta x_{41}^{\dagger}x_{42}\equiv0\pmod{\delta^{3}}$.
We then have $\omega^{-a}\omega^{b}+\omega^{-c}\omega^{d}\equiv0\pmod{\delta^{3}}$, so by Lemma 6 we have $b-a\equiv d-c\pmod{4}$, and we may take $x\equiv c-a\pmod{4}$.\\
\end{proof}
Now that we know such an $x$ exists, without loss of generality we have
\[
\rho_{3}^{k}(\omega^{x}_{[1]}U\omega^{-c}_{[1]}\omega^{-d}_{[2]})=\begin{pmatrix}1 & 1 & 0 & 0\\
1 & 1 & 0 & 0\\
0 & 0 & 0 & 0\\
0 & 0 & 0 & 0
\end{pmatrix}+\delta\begin{pmatrix}0 & 0 & y_{13} & y_{14}\\
0 & 0 & y_{23} & y_{24}\\
y_{31} & y_{32} & y_{33} & y_{34}\\
y_{41} & y_{42} & y_{43} & y_{44}
\end{pmatrix}+\delta^{2}\begin{pmatrix}0 & 0 & z_{13} & z_{14}\\
0 & 0 & z_{23} & z_{24}\\
z_{31} & z_{32} & z_{33} & z_{34}\\
z_{41} & z_{42} & z_{43} & z_{44}
\end{pmatrix}
\]
using the notation presented in Remark 6.

\newtheorem{remarkwlog}[remarkdomain]{Remark}
\begin{remarkwlog}
\rm{We note that the proof for Lemma 13 only used the first two columns, and did not make any assumptions about columns three or four.}
\end{remarkwlog}

\newtheorem{lemma4x4equal}[lemma1]{Lemma}
\begin{lemma4x4equal}
We always have $y_{13}=y_{23}$, $y_{14}=y_{24}$, $z_{13}=z_{23}$, and $z_{14}=z_{24}$.
\end{lemma4x4equal}
\begin{proof}
Taking the inner products of columns two and three, we get $\delta(y_{13}+y_{23})+\delta^{\dagger}\delta(y_{32}^{\dagger}y_{33}+y_{42}^{\dagger}y_{43})+\delta^{2}(z_{13}+z_{23})\equiv0\pmod{\delta^{3}}$.
This implies that $y_{13}=y_{23}$.$ $ Taking
the inner product of the second column with itself, we get that $\delta^{\dagger}\delta(y_{32}^{\dagger}y_{32}+y_{42}^{\dagger}y_{42})\equiv0\pmod{\delta^{3}}$,
so $y_{32}=y_{42}$. Taking the inner product
of the third column and itself, we get that $\delta^{\dagger}\delta(y_{13}^{\dagger}y_{13}+y_{23}^{\dagger}y_{23}+y_{33}^{\dagger}y_{33}+y_{43}^{\dagger}y_{43})\equiv0\pmod{\delta^{3}}$.
Because $y_{13}=y_{23}$, this implies that $y_{33}= y_{43}$.
Using this fact, the inner product of columns two and three now gives
us that $\delta^{2}(z_{13}+z_{23})\equiv0\pmod{\delta^{3}}$, so
we must have $z_{13}=z_{23}$. A similar argument
gives that $y_{14}=y_{24}$ and $z_{14}=z_{24}$.
\end{proof}

We now have
\[
\rho_{3}^{k}(\omega^{x}_{[1]}U\omega^{-c}_{[1]}\omega^{-d}_{[2]})=\begin{pmatrix}1 & 1 & 0 & 0\\
1 & 1 & 0 & 0\\
0 & 0 & 0 & 0\\
0 & 0 & 0 & 0
\end{pmatrix}+\delta\begin{pmatrix}0 & 0 & y_{13} & y_{14}\\
0 & 0 & y_{13} & y_{14}\\
y_{31} & y_{32} & y_{33} & y_{34}\\
y_{41} & y_{42} & y_{43} & y_{44}
\end{pmatrix}+\delta^{2}\begin{pmatrix}0 & 0 & z_{13} & z_{14}\\
0 & 0 & z_{13} & z_{14}\\
z_{31} & z_{32} & z_{33} & z_{34}\\
z_{41} & z_{42} & z_{43} & z_{44}
\end{pmatrix},
\]
\noindent
by Corollary 3, we can apply the row operation $H_{[1,2]}$ and we are done.

\subsubsection{Case (ii)}
For case (ii), we have
\[
\rho_{3}^{k}(U)=\begin{pmatrix}\omega^{a} & \omega^{b} & \omega^{c} & \omega^{d}\\
\omega^{e} & \omega^{f} & \omega^{g} & \omega^{h}\\
\delta x_{31} & \delta x_{32} & \delta x_{33} & \delta x_{34}\\
\delta x_{41} & \delta x_{42} & \delta x_{43} & \delta x_{44}
\end{pmatrix}.
\]
\newtheorem{lemmacaseii}[lemma1]{Lemma}
\begin{lemmacaseii}
If $U$ is of case (ii), there always exists an $x\pmod{4}$ such that $a+x\equiv e\pmod{4}$,
$b+x\equiv f\pmod{4}$, $c+x\equiv g\pmod{4}$, and $d+x\equiv h\pmod{4}$.
\end{lemmacaseii}

\begin{proof}
By Remark 7, we can proceed as in the proof of Lemma 13, applied to the first and second columns, to show that $e-a\equiv f-b\pmod{4}$. Using a similar argument on the other columns of the matrix, it can be shown that $e-a\equiv f-b\equiv g-e\equiv h-d\pmod{4}$. Setting $x\equiv e-a\pmod{4}$, we are done.
\end{proof}

Now that we know such an $x$ exists, we have
\[
\rho_{3}^{k}(\omega^{x}_{[1]}U)=\begin{pmatrix}\omega^{e} & \omega^{f} & \omega^{g} & \omega^{h}\\
\omega^{e} & \omega^{f} & \omega^{g} & \omega^{h}\\
\delta x_{31} & \delta x_{32} & \delta x_{33} & \delta x_{34}\\
\delta x_{41} & \delta x_{42} & \delta x_{43} & \delta x_{44}
\end{pmatrix},
\]
\noindent
and by Corollary 3, we can apply the row operation $H_{[1,2]}$ and we are done.
\subsubsection{Case (iii)}
In case (iii), we have
\[
\rho_{3}^{k}(U)=\begin{pmatrix}\omega^{a} & \omega^{b} & \delta x_{13} & \delta x_{14}\\
\omega^{c} & \omega^{d} & \delta x_{23} & \delta x_{24}\\
\delta x_{31} & \delta x_{32} & \omega^{e} & \omega^{f}\\
\delta x_{41} & \delta x_{42} & \omega^{g} & \omega^{h}
\end{pmatrix}.
\]

\newtheorem{lemmacaseiv}[lemma1]{Lemma}
\begin{lemmacaseiv}
If $U$ is of case (iii), there exists an $x\pmod{4}$ such that $a+x\equiv c\pmod{4}$ and such that $b+x\equiv d\pmod{4}$.
\end{lemmacaseiv}

\begin{proof}
By Remark 7, we can proceed as in the proof of Lemma 13.
\end{proof}

Now that we know such an $x$ exists, we have
\[
\rho_{3}^{k}(\omega^{x}_{[1]}U)=\begin{pmatrix}\omega^{c} & \omega^{d} & \delta y_{13} & \delta y_{14}\\
\omega^{c} & \omega^{d} & \delta x_{23} & \delta x_{24}\\
\delta x_{31} & \delta x_{32} & \omega^{e} & \omega^{f}\\
\delta x_{41} & \delta x_{42} & \omega^{g} & \omega^{h}
\end{pmatrix},
\]
where $y_{13}=\omega^{x}x_{13}$ and $y_{14}=\omega^{x}x_{14}$.

\newtheorem{lemmaivequal}[lemma1]{Lemma}
\begin{lemmaivequal}
We always have $\delta y_{13}\equiv\delta x_{23}\pmod{\delta^{2}}$ and $\delta y_{14}\equiv\delta x_{24}\pmod{\delta^{2}}$.
\end{lemmaivequal}

\begin{proof}
Taking the inner product of the third column and itself, we have $\delta^{\dagger}\delta y^{\dagger}_{13}y_{13}+\delta^{\dagger}\delta x^{\dagger}_{23}x_{23}+\omega^{-e}\omega^{e}+\omega^{-g}\omega^{g}\equiv0\pmod{\delta^{3}}$. This implies that $y^{\dagger}_{13}y_{13}\equiv x^{\dagger}_{23}x_{23}\pmod{\delta}$, so $y_{13}\equiv x_{23}\pmod{\delta}$, thus $\delta y_{13}\equiv\delta x_{23}\pmod{\delta^{2}}$. Similarly, $\delta y_{14}\equiv\delta x_{24}\pmod{\delta^{2}}$.
\end{proof}

By Corollary 2, applying an elementary Hadamard operation to the first two rows will not increase its $\delta$-exponent, thus applying the operation will reduce the matrix to case (i) or case (ii). Because both of these cases have been shown to satisfy Lemma 9 above, case (iii) also satisfies Lemma 9.
\subsubsection{Case (iv)}
In case (iv), we have
\[
\rho_{3}^{k}(U)=\begin{pmatrix}\omega^{a} & \omega^{b} & \delta x_{13} & \delta x_{14}\\
\omega^{c} & \omega^{d} & \delta x_{23} & \delta x_{24}\\
\omega^{e} & \omega^{f} & \omega^{g} & \omega^{h}\\
\omega^{j} & \omega^{l} & \omega^{m} & \omega^{p}
\end{pmatrix}.
\]
\newtheorem{lemmacaseiii}[lemma1]{Lemma}
\begin{lemmacaseiii}
If $U$ is of case (iv), there exists an $x\pmod{4}$ such that $a+x\equiv c\pmod{4}$ and such that $b+x\equiv d\pmod{4}$.
\end{lemmacaseiii}

\begin{proof}
Using an analogous proof to Lemma 13, applied to the first and second rows, we get $b-a \equiv d-c\pmod{4}$. We can then take $x\equiv c-a\pmod{4}$.
\end{proof}
\noindent
Now that we know such an $x$ exists, we have
\[
\rho_{3}^{k}(\omega^{x}_{[1]}U)=\begin{pmatrix}\omega^{c} & \omega^{d} & \delta y_{13} & \delta y_{14}\\
\omega^{c} & \omega^{d} & \delta x_{23} & \delta x_{24}\\
\omega^{e}& \omega^{f}& \omega^{g} & \omega^{h}\\
\omega^{j} & \omega^{l} & \omega^{m} & \omega^{p}
\end{pmatrix},
\]
\noindent
where $y_{13}=\omega^{x}x_{13}$ and $y_{14}=\omega^{x}x_{14}$.

\newtheorem{lemmaiiiequal}[lemma1]{Lemma}
\begin{lemmaiiiequal}
We always have $\delta y_{13}\equiv\delta x_{23}\pmod{\delta^{2}}$ and $\delta y_{14}\equiv\delta x_{24}\pmod{\delta^{2}}$.
\end{lemmaiiiequal}

\begin{proof}
Taking the inner product of the third column and itself, we have $\delta^{\dagger}\delta y^{\dagger}_{13}y_{13}+\delta^{\dagger}\delta x^{\dagger}_{23}x_{23}+\omega^{-g}\omega^{g}+\omega^{-m}\omega^{m}\equiv0\pmod{\delta^{3}}$. This implies that $y^{\dagger}_{13}y_{13}\equiv x^{\dagger}_{23}x_{23}\pmod{\delta}$, so $y_{13}\equiv x_{23}\pmod{\delta}$, thus $\delta y_{13}\equiv\delta x_{23}\pmod{\delta^{2}}$. Similarly, $\delta y_{14}\equiv\delta x_{24}\pmod{\delta^{2}}$.
\end{proof}

If $\delta y_{13}\equiv\delta x_{23}\pmod{\delta^{3}}$ and $\delta y_{14}\equiv\delta x_{24}\pmod{\delta^{3}}$ then we are done, as applying an elementary Hadamard operation to the first two rows of $\omega^{x}_{[1]}U$ will reduce their $\delta$-exponent by Corollary 3 and will reduce $\omega^{x}_{[1]}U$ to case (i). Assume therefore, without loss of generality, that $\delta y_{13}\not\equiv\delta x_{23}\pmod{\delta^{3}}$. Since $\delta y_{13}+\delta x_{23}
\equiv 0\pmod{\delta^2}$ and $\delta y_{13}+ \delta x_{23} \not\equiv
0 \pmod{\delta^3}$, we therefore have $\delta y_{13}+ \delta x_{23}
\equiv \delta^2 \pmod{\delta^3}$.

\newtheorem{lemmaiii34}[lemma1]{Lemma}
\begin{lemmaiii34}
For $\omega^{x}_{[1]}U$, there exists some $y\pmod{4}$ such that $e+y\equiv j\pmod{4}$ and such that $f+y\equiv l\pmod{4}$.
\end{lemmaiii34}

\begin{proof}
Taking the inner product of the first and second columns, a we get $\omega^{-c}\omega^{d}+\omega^{-c}\omega^{d}+\omega^{-e}\omega^{f}+\omega^{-j}\omega^{l}\equiv0\pmod{\delta^{3}}$. This implies that $\omega^{-e}\omega^{f}+\omega^{-j}\omega^{l}\equiv0\pmod{\delta^{3}}$, so by Lemma 6 we can take $y\equiv j-e\pmod{4}$.
\end{proof}
\noindent
Now we have
\[
\rho_{3}^{k}(\omega^{x}_{[1]}\omega^{y}_{[3]}U)=\begin{pmatrix}\omega^{c} & \omega^{d} & \delta y_{13} & \delta y_{14}\\
\omega^{c} & \omega^{d} & \delta x_{23} & \delta x_{24}\\
\omega^{j}& \omega^{l}& \omega^{g+y} & \omega^{h+y}\\
\omega^{j} & \omega^{l} & \omega^{m} & \omega^{p}
\end{pmatrix}.
\]

\newtheorem{lemmaiii342}[lemma1]{Lemma}
\begin{lemmaiii342}
For $\omega^{x}_{[1]}\omega^{y}_{[3]}U$, we always have $g+y-m\equiv2\pmod{4}$ and $h+y-p\equiv2\pmod{4}$.
\end{lemmaiii342}

\begin{proof}
Taking the inner product of the second and third columns, we have $\omega^{-d}(\delta y_{13} + \delta x_{23})+\omega^{g+y-l}+\omega^{m-l}\equiv0\pmod{\delta^{3}}$.  Since $\delta y_{13} + \delta x_{23} \equiv \delta^2\pmod{\delta^{3}}$, and noting that $\omega^{-d}\delta^2\equiv \delta^2\pmod{\delta^{3}}$, the inner product becomes $\delta^2 + \omega^{g+y-l} + \omega^{m-l} \equiv 0 \pmod{\delta^{3}}$. We then have that $\omega^{-l}(\omega^{g+y}+\omega^{m})\equiv\delta^{2}\equiv1+\omega^{2}\pmod{\delta^{3}}$. Multiplying both sides through by $\omega^{l-m}$, we get that $\omega^{g+y-m}\equiv\omega^{2}\pmod{\delta^{3}}$. Then, by Lemma 6, we must have $g+y-m\equiv 2\pmod{4}$. A similar proof shows $h+y-p\equiv2\pmod{4}$.
\end{proof}

Because $g+y-m\equiv2\pmod{4}$ and $h+y-p\equiv2\pmod{4}$, the third and fourth rows of $\omega^{x}_{[1]}\omega^{y}_{[3]}U$ are congruent modulo $\delta^{2}$. By Corollary 2, applying an elementary Hadamard operation to the third and fourth rows of $\omega^{x}_{[1]}\omega^{y}_{[3]}U$ will not increase their $\delta$-exponent, thus applying the operation will reduce the matrix to case (i) or case (iii). Because both of these cases have been shown to satisfy Lemma 9 above, case (iv) also satisfies Lemma 9.

\subsubsection{Case (v)}
\tab In case (v), we have
\[
\rho_{3}^{k}(U)=\begin{pmatrix}\omega^{a} & \omega^{b} & \omega^{c} & \omega^{d}\\
\omega^{e} & \omega^{f} & \omega^{g} & \omega^{h}\\
\omega^{j} & \omega^{l} & \omega^{m} & \omega^{p}\\
\omega^{q} & \omega^{r} & \omega^{s} & \omega^{t}
\end{pmatrix}.
\]
Taking the inner products of the first and second rows, we get $\omega^{e-a}+\omega^{f-b}+\omega^{g-c}+\omega^{h-d}\equiv0\pmod{\delta^{3}}$. Lemma 7 gives us three possibilities for the values of these exponents.\\
\textbf{Case 1.} If $e-a\equiv f-b\equiv g-c\equiv h-d\pmod{4}$, we can set $x\equiv e-a\pmod{4}$. Then the first two rows of $\omega^{x}_{[1]}U$ are equivalent modulo $\delta^{3}$, so we are done by Corollary 3.\\
\textbf{Case 2.} Assume that $e-a$, $f-b$, $g-c$, and $h-d$ are distinct modulo $4$. Then, without loss of generality, we have
\[
\rho_{3}^{k}(U)=\begin{pmatrix}1 & 1 & 1 & 1\\
1 & \omega & \omega^{2} & \omega^{3}\\
1 & \omega^{l} & \omega^{m} & \omega^{p}\\
\omega^{q} & \omega^{r} & \omega^{s} & \omega^{t}
\end{pmatrix}.
\]

Taking the inner product of the first and third rows, we get $1+\omega^{l}+\omega^{m}+\omega^{p}\equiv0\pmod{\delta^{3}}$. Taking the inner product of the second and third rows, we get $1+\omega^{l-1}+\omega^{m-2}+\omega^{p-3}\equiv0\pmod{\delta^{3}}$. By Lemma 7, there are 3 options for the values of $l$, $m$, and $p$.\\

\textbf{Case 2.1.} Assume that $0$, $l$, $m$, and $p$ are distinct modulo $4$. Then, there are 6 possible permutations of $l$, $m$, and $p$. Note that this covers case (c) of Lemma 7.
\begin{itemize}
\item{If $(l,\:m,\:p)\equiv (1,\:2,\:3)\pmod{4}$, then the second and third rows of $U$ are congruent modulo $\delta^{3}$. By Corollary 3, applying an elementary Hadamard operation reduces the $\delta$-exponent of the rows and this case reduces to case (ii).}
\item{If $(l,\:m,\:p)\equiv (1,\:3,\:2)\pmod{4}$, then the inner product of the second and third rows becomes $1+1+\omega+\omega^{-1}\not\equiv0\pmod{\delta^{3}}$, thus this cannot occur.}
\item{If $(l,\:m,\:p)\equiv (2,\:1,\:3)\pmod{4}$, then the inner product of the second and third rows becomes $1+\omega+\omega^{-1}+1\not\equiv0\pmod{\delta^{3}}$, thus this cannot occur.}
\item{If $(l,\:m,\:p)\equiv (2,\:3,\:1)\pmod{4}$, then the inner product of the second and third rows becomes $1+\omega+\omega+\omega^{-2}\not\equiv0\pmod{\delta^{3}}$, thus this cannot occur.}
\item{If $(l,\:m,\:p)\equiv (3,\:1,\:2)\pmod{4}$, then the inner product of the second and third rows becomes $1+\omega^{2}+\omega^{-1}+\omega^{-1}\not\equiv0\pmod{\delta^{3}}$, thus this cannot occur.}
\item{If $(l,\:m,\:p)\equiv (3,\:2,\:1)\pmod{4}$, then the  second and third rows of $U$ are congruent modulo $\delta^{2}$ and by Corollary 2 we can apply an elementary Hadamard gate to these rows which will reduce this case to either case (ii) or case (iv) without increasing $k$.}
\end{itemize}

\textbf{Case 2.2.} Assume that at least one of $l$, $m$, and $p$ is congruent to $0$ modulo $4$, and the other two must be congruent modulo $4$. Note that this covers both cases (a) and (b) of Lemma 7.
\newtheorem{casevequal}[lemma1]{Lemma}
\begin{casevequal}
Two rows of $U$ are congruent modulo $\delta^{2}$.
\end{casevequal}
\begin{proof}
\begin{itemize}
\item{Assume that $l\equiv0\pmod{4}$ and $m\equiv p$. Then, if $m\equiv p\equiv 0$ or $m\equiv p\equiv 2$ we are done, as the first and third rows will be equivalent modulo $\delta^{2}$. Assume that $m\equiv p\equiv 1$. Then the inner product of the second and third rows will be $1+\omega^{-1}+\omega^{-1}+\omega^{-2}\equiv1+\omega^{2}\not\equiv0\pmod{\delta^{3}}$, thus this cannot occur. Next, assume $m\equiv p\equiv 3$. Then the inner product of the second and third row will be $1+\omega^{-1}+\omega+1\equiv\omega+\omega^{3}\not\equiv0\pmod{\delta^{3}}$, thus this cannot occur.}
\item{Now, assume that $m\equiv0\pmod{4}$ and $l\equiv p$. If $l\equiv p$ is odd, the third row will be equivalent to the second row modulo $\delta^{2}$. If $l\equiv p$ is even, the first and third rows will be equivalent modulo $\delta^{2}$.}
\item{Lastly, assume that $p\equiv0\pmod{4}$ and $l\equiv m$. If $l\equiv m\equiv 0$ or $l\equiv m\equiv 2$ we are done, as the first and third rows will be congruent modulo $\delta^{2}$. Assume that $l\equiv m\equiv 1$. Then the inner product of the second and third rows will be $1+1+\omega^{-1}+\omega^{-3}\equiv\omega+\omega^{3}\not\equiv0\pmod{\delta^{3}}$, thus this cannot occur. Next, assume that $l\equiv m\equiv 3$. Then the inner product of the second and third row will be $1+\omega^{2}+\omega+\omega^{-3}\equiv1+\omega^{2}\not\equiv0\pmod{\delta^{3}}$, thus this cannot occur, thus we must always have two rows of $U$ congruent modulo $\delta^{2}$.}
\end{itemize}
\end{proof}

By the above lemma and Corollary 2, there always exist two rows of $U$ to which we can apply an elementary Hadamard operation without increasing $k$. This will reduce $U$ to an instance of case (ii) or (iv), which have already been shown to satisfy Lemma 9. By Lemma 7, we have satisfied all possibilites for case 2.\\
\textbf{Case 3.} The last case given by Lemma 7 occurs without loss of generality when $e-a\equiv f-b\pmod{4}$, $g-c\equiv h-d\pmod{4}$, and $e-a\not\equiv g-c\pmod{4}$. We can let $a\equiv b\equiv c\equiv d\equiv e\equiv f\equiv j\equiv 0\pmod{4}$ because we know such row and column operations exist. This gives us:
\[
\rho_{3}^{k}(U)=\begin{pmatrix}1 & 1 & 1 & 1\\
1 & 1 & \omega^{g} & \omega^{g}\\
1 & \omega^{l} & \omega^{m} & \omega^{p}\\
\omega^{q} & \omega^{r} & \omega^{s} & \omega^{t}
\end{pmatrix}
\]
\textbf{Case 3.1.} If $g$ is even modulo $4$, we are done, as the first two rows of U will be congruent modulo $\delta^{2}$. By Corollary 2, applying an elementary Hadamard operation will not increase the least $\delta$-exponent $k$ of $U$ and this case will reduce to an instance of case (ii) or (iv).\\
\textbf{Case 3.2.} Assume $g\equiv 1\pmod{4}$. Then, we have
\[
\rho_{3}^{k}(U)=\begin{pmatrix}1 & 1 & 1 & 1\\
1 & 1 & \omega & \omega\\
1 & \omega^{l} & \omega^{m} & \omega^{p}\\
\omega^{q} & \omega^{r} & \omega^{s} & \omega^{t}
\end{pmatrix}.
\]
\newtheorem{lemmavlast1}[lemma1]{Lemma}
\begin{lemmavlast1}
Two rows of $U$ are congruent modulo $\delta^{2}$.
\end{lemmavlast1}
\begin{proof}
Taking the inner product of the first and third rows, we get $1+\omega^{l}+\omega^{m}+\omega^{p}\equiv0\pmod{\delta^{3}}$. Taking the inner product of the second and third rows, we get $1+\omega^{l}+\omega^{m-1}+\omega^{p-1}\equiv0\pmod{\delta^{3}}$. By Lemma 7, either $0$, $l$, $m$, and $p$ are distinct modulo $4$, or at least one of $l$, $m$, and $p$ is congruent to $0$ and the other two are congruent to each other modulo $4$.\\
\textbf{Case 3.2.1.} Assume $0$, $l$, $m$, and $p$ are distinct modulo $4$.

\begin{itemize}
\item{If $(l,\:m,\:p)\equiv (1,\:2,\:3)\pmod{4}$, then the inner product of the second and third rows becomes $1+\omega+\omega+\omega^{2}\not\equiv0\pmod{\delta^{3}}$, thus this cannot occur.}
\item{If $(l,\:m,\:p)\equiv (1,\:3,\:2)\pmod{4}$, then the inner product of the second and third rows becomes $1+\omega+\omega^{2}+\omega\not\equiv0\pmod{\delta^{3}}$, thus this cannot occur.}
\item{If $(l,\:m,\:p)\equiv (2,\:1,\:3)\pmod{4}$, then the second and third rows of $U$ are congruent modulo $\delta^{2}$, and applying an elementary Hadamard operation to them will reduce this case to either case (ii) or case (iv) without increasing $k$ by Corollary 2.}
\item{If $(l,\:m,\:p)\equiv (2,\:3,\:1)\pmod{4}$, then the second and third rows of $U$ are congruent modulo $\delta^{2}$, and applying an elementary Hadamard operation to them will reduce this case to either case (ii) or case (iv) without increasing $k$ by Corollary 2.}
\item{If $(l,\:m,\:p)\equiv (3,\:1,\:2)\pmod{4}$, then the inner product of the second and third rows becomes $1+\omega^{3}+1+\omega\not\equiv0\pmod{\delta^{3}}$, thus this cannot occur.}
\item{If $(l,\:m,\:p)\equiv (3,\:2,\:1)\pmod{4}$, then the inner product of the second and third rows becomes $1+\omega^{3}+\omega+1\not\equiv0\pmod{\delta^{3}}$, thus this cannot occur.}
\end{itemize}

\textbf{Case 3.2.2.} The other case occurs when at least one of $l$, $m$, and $p$ is congruent to $0$ and the other two are congruent to each other modulo $4$.
\begin{itemize}
\item{Assume $l\equiv0\pmod{4}$ and $m\equiv p\pmod{4}$. Then, if $m\equiv p\pmod{4}$ is even, the first and third rows of $U$ are congruent modulo $\delta^{2}$, and if $m\equiv p\pmod{4}$ is odd, the second and third rows of $U$ are congruent modulo $\delta^{2}$.}
\item{Assume that $m\equiv0\pmod{4}$ and $l\equiv p\pmod{4}$. Then, if $l\equiv p \pmod{4}$ is even, we are done, because the first and third rows are congruent modulo $\delta^{2}$. If $l\equiv p\equiv1\pmod{4}$, then the inner product of the second and third rows will be $1+\omega+\omega^{-1}+1\equiv\omega+\omega^{3}\not\equiv0\pmod{\delta^{3}}$, thus this cannot occur. If $l\equiv p\equiv3\pmod{4}$, then the inner product of the second and third rows will be $1+\omega^{3}+\omega^{-1}+\omega^{2}\equiv1+\omega^{2}\not\equiv0\pmod{\delta^{3}}$, thus this case cannot occur.}
\item{Lastly, assume $p\equiv0\pmod{4}$ and $l\equiv m\pmod{4}$. This case is analogous to the previous case where $m\equiv0\pmod{4}$.}
\end{itemize}
\end{proof}

Because two rows of $U$ are always congruent modulo $\delta^{2}$, by Corollary 2 we can apply an elementary Hadamard gate to these rows without increasing $k$. Applying this operation to $U$ will reduce this case to case (ii) or case (iv), which have been shown to satisfy Lemma 9.\\
\textbf{Case 3.3.} The last instance of this case occurs when $g\equiv 3\pmod{4}$. Then, we have
\[
\rho_{3}^{k}(U\omega_{[3]}\omega_{[4]})=\begin{pmatrix}1 & 1 & \omega & \omega\\
1 & 1 & 1 & 1\\
1 & \omega^{l} & \omega^{m+1} & \omega^{p+1}\\
\omega^{q} & \omega^{r} & \omega^{s+1} & \omega^{t+1}
\end{pmatrix},
\]
and the proof becomes analogous to the proof for case 3.2.

\subsection{Statement of the main theorem}
\newtheorem{main}{Theorem}
\begin{main}
Given a unitary $n\times n$ matrix $U$ with $n\le 4$ and entries in $\mathbb{D}[\omega]$ and least $\delta$-exponent $k$, then $U$ can be written as a product of at most $Nk+M$ elementary operators, where $M$ and $N$ are the bounds described in Lemmas 8 and 9 respectively.
\end{main}
\begin{proof}
By induction on $k$. If $k=0$, by Lemma 8 there exists a sequence $S_{1},\ldots,S_{m}$ of elementary operators, of length at most $M$, such that $S_{1}\ldots S_{m}U=I$. We can then write $U=S^{-1}_{m}\ldots S^{-1}_{1}$. By Corollary 1, $k=1$ does not occur. For $k>1$, by Lemma 9, there exist two sequences $S_{1},\ldots,S_{m}$ and $T_{1},\ldots,T_{j}$ of elementary operators, of total length at most $N$, such that $U'=S_{1}\ldots S_{m}UT_{1}\ldots T_{j}$ has least $\delta$-exponent $k'<k$. Applying the induction hypothesis, we can then write $U'$ as a product $R_{1}\ldots R_{l}$ of elementary operators of length at most $Nk'+M$. The claim then follows by writing
\[
U=S^{-1}_{m}\ldots S^{-1}_{1}R_{1}\ldots R_{l}T^{-1}_{j}\ldots T^{-1}_{1}.
\]
\end{proof}

\section{Synthesis of $1$- and $2$-qubit Clifford+$T$ circuits}

\subsection{Main synthesis result}
Let $U$ be a $2^{n} \times 2^{n}$ unitary matrix with $n \le 2$ and entries in $\mathbb{D}[\omega]$. By Theorem 1, $U$ can be decomposed into elementary matrices of types $H_{[j,m]}$, $X_{[j,m]}$, and $\omega_{[j]}$. It is known that these matrices can be further decomposed into controlled-not gates and multiply-controlled $X$, $H$, $T$, and $\omega$-gates, for example using Gray codes [1, Sec 4.5.2], and all of these gates have well known representations in Clifford+$T$ circuits.

\subsection{One ancilla is sufficient}
It remains to show that $U$ can be represented as a circuit with at most one ancilla. It is known that for $n>1$, one ancilla is sometimes necessary [4].\\
\tab Giles and Selinger [2] showed that one could implement a multiply-controlled $X$-gate, $H$-gate, and $T$-gate using one ancilla in the following way:
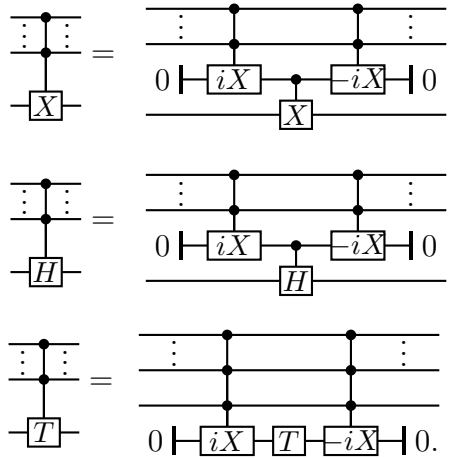
\[
  \m{\begin{qcircuit}[scale=0.47]
    \grid{2}{0.5,2,3};
    \draw(0.4,2.7) node[anchor=center]{$\vdots$};
    \controlled{\gate{$X$}}{1,0.5}{2,3};
    \draw(1.6,2.7) node[anchor=center]{$\vdots$};
  \end{qcircuit}}
  =\!\!
  \m{\begin{qcircuit}[scale=0.47]
    \grid{8.5}{0,2,3};
    \gridx{1}{7.5}{1};
    \draw(1,2.7) node[anchor=center]{$\vdots$};
    \init{$0$}{1,1};
    \controlled{\widegate{$iX$}{.75}}{2.5,1}{2,3};
    \controlled{\gate{$X$}}{4.25,0}{1};
    \controlled{\widegate{$-iX$}{.75}}{6,1}{2,3};
    \term{$0$}{7.5,1};
    \draw(7.5,2.7) node[anchor=center]{$\vdots$};
  \end{qcircuit}}
\]\[
  \m{\begin{qcircuit}[scale=0.47]
    \grid{2}{0.5,2,3};
    \draw(0.4,2.7) node[anchor=center]{$\vdots$};
    \controlled{\gate{$H$}}{1,0.5}{2,3};
    \draw(1.6,2.7) node[anchor=center]{$\vdots$};
  \end{qcircuit}}
  =\!\!
  \m{\begin{qcircuit}[scale=0.47]
    \grid{8.5}{0,2,3};
    \gridx{1}{7.5}{1};
    \draw(1,2.7) node[anchor=center]{$\vdots$};
    \init{$0$}{1,1};
    \controlled{\widegate{$iX$}{.75}}{2.5,1}{2,3};
    \controlled{\gate{$H$}}{4.25,0}{1};
    \controlled{\widegate{$-iX$}{.75}}{6,1}{2,3};
    \term{$0$}{7.5,1};
    \draw(7.5,2.7) node[anchor=center]{$\vdots$};
  \end{qcircuit}}
\]\[
  \m{\begin{qcircuit}[scale=0.47]
    \grid{2}{0.5,2,3};
    \draw(0.4,2.7) node[anchor=center]{$\vdots$};
    \controlled{\gate{$T$}}{1,0.5}{2,3};
    \draw(1.6,2.7) node[anchor=center]{$\vdots$};
  \end{qcircuit}}
  =\!\!\!
  \m{\begin{qcircuit}[scale=0.47]
    \grid{8.5}{1,2,3};
    \gridx{1}{7.5}{0};
    \draw(1,2.7) node[anchor=center]{$\vdots$};
    \init{$0$}{1,0};
    \controlled{\widegate{$iX$}{.75}}{2.5,0}{1,2,3};
    \gate{$T$}{4.25,0}{1};
    \controlled{\widegate{$-iX$}{.75}}{6,0}{1,2,3};
    \term{$0$.}{7.5,0};
    \draw(7.5,2.7) node[anchor=center]{$\vdots$};
  \end{qcircuit}}
\]

Combining this result with Theorem 1 clearly shows that $U$ can be implemented as a circuit using at most one ancilla.

\subsection{Complexity}

The proof of Theorem 1 for the $2^{n} \times 2^{n}$ case with $0\le n \le 2$ above gives an algorithm for synthesizing a Clifford+$T$ circuit with ancillas from a given operator $U$. The algorithm yields a circuit of $O(k)$ gates, where $k$ is the least $\delta$-exponent of $U$.\\
\tab Because we only prove this algorithm for $0\le n \le 2$, it does not make sense to express the complexity of the algorithm in terms of $n$.

\section{Conclusion and future work}
In conclusion, we have presented an efficient algorithm for the exact synthesis of Clifford+$T$ circuits from given $1$- and $2$-qubit operators using decomposition into elementary operators.\\
\tab This paper only gives a proof for the algorithm in the $1$ and $2$ qubit cases. Future work will involve extending this to the $n$ qubit case. If our methods can be extended to $n>2$, one may hope to prove that the complexity scales as $O(b^{n}nk)$, similarly to the algorithm presented in [3]. It is an open question whether this is the case, and whether $b<4$.\\
\tab It would also be interesting to carry out a practical comparison between our algorithm and that of [3], to determine the constants hidden in the big-O notation of each algorithm's complexity. Such a comparison could be meaningful even for fixed $n$.
\section{Acknowledgements}
I wish to thank my supervisor, Dr. Peter Selinger for his guidance and help throughout this project. I would also like to thank my family, friends, and colleagues who made my time at Dalhousie University a memorable experience.



\begin{thebibliography}{99}
\bibitem[1]{1}
M. A. Nielsen and I. L. Chuang. Quantum Computation and Quantum \newblock Information. Cambridge University Press, 2002.

\bibitem[2]{2}
B. Giles and P. Selinger. Exact synthesis of multiqubit \newblock Clifford+$T$ circuits. Physical Review A, 87:032332, Mar 2013. Also available from arXiv:1212.0506, doi:10.1103/PhysRevA.87.032332.

\bibitem[3]{3}
V. Kliuchnikov. Synthesis of unitaries with Clifford+$T$ \newblock circuits. arXiv: 1306.3200, June 2013.

\bibitem[4]{4}
V. Kliuchnikov, D. Maslov, and M. Mosca. Fast and efficient \newblock exact synthesis of single qubit unitaries generated \newblock by Clifford and $T$ gates. arXiv:1206.5236v2, June 2012.

\end{thebibliography}
\end{document}